\documentclass[submission]{eptcs}
\providecommand{\event}{QPL 2014}
\usepackage{breakurl}

\usepackage{graphicx}
\usepackage{mathptmx}
\usepackage{amssymb}
\usepackage{amsmath}
\usepackage{xfrac}
\usepackage{bbold}
\usepackage{setspace}
\usepackage{multirow}
\usepackage{amsthm}

\newcommand{\M}{\mathcal{M}}
\newcommand{\E}{\mathcal{E}}
\newcommand{\ket}[1]{\left| #1 \right>}

\def\Forall#1{\forall \; {#1}\boldsymbol{.}\;}

\theoremstyle{plain}
\newtheorem{theorem}{Theorem}[section]
\newtheorem{proposition}[theorem]{Proposition}
\newtheorem{corollary}[theorem]{Corollary}

\theoremstyle{definition}
\newtheorem{definition}[theorem]{Definition}
\newtheorem{example}[theorem]{Example}

\title{Reflections on the PBR Theorem: \\ Reality Criteria \& Preparation Independence}
\def\titlerunning{Reality Criteria \& Preparation Independence}
\author{
Shane Mansfield
\institute{Quantum Group \\ Department of Computer Science \\ University of Oxford}
\email{shane.mansfield@cs.ox.ac.uk}
}
\def\authorrunning{Shane Mansfield}
\date{\today}

\begin{document}

\maketitle

\begin{abstract}
This paper contains initial work on attempting to bring recent developments in the foundations of quantum mechanics concerning the nature of the wavefunction within the scope of more logical and structural methods. A first step involves dualising a criterion for the reality of the wavefunction proposed by Harrigan \& Spekkens, which was central to the Pusey-Barrett-Rudolph theorem. The resulting criterion has several advantages, including the avoidance of certain technical difficulties relating to sets of measure zero. By considering the `reality' not of the wavefunction but of the observable properties of any ontological physical theory a new characterisation of non-locality and contextuality is found. Secondly, a careful analysis of preparation independence, one of the key assumptions of the PBR theorem, leads to a precise analogy with the kind of locality prohibited by Bell's theorem. Motivated by this, we propose a weakening of the assumption to something analogous to no-signalling. This amounts to allowing global or non-local correlations in the joint ontic state, which nevertheless do not allow for superluminal signalling. This is, at least, consistent with the Bell and Kochen-Specker theorems. We find a counter-example to the PBR argument, which violates preparation independence, but does satisfy this physically motivated assumption. The question of whether the PBR result can be strengthened to hold under the relaxed assumption is therefore posed.
\end{abstract}

\section{Introduction}

The issue of the reality of the wavefunction has received a lot of attention recently (see especially \cite{pusey:12,colbeck:12,hardy:13a}). In this paper, we show that insights may also be gained by taking a similar approach to considering the `reality' of objects and properties in physical theories more generally, and in particular that such an approach can provide a new perspective on non-locality and contextuality.
The first step will be to formalise a suitably general criterion for `reality' inspired by the Harrigan-Spekkens criterion for the reality of the wavefunction \cite{harrigan:10}, which was the subject of the Pusey-Barrett-Rudolph theorem \cite{pusey:12}.

The aim is to formulate the ideas in a manner that can allow for a deeper, structural understanding of what is at play. Indeed, the initial motivation was to bring considerations of this kind within the scope of the methods of the unified sheaf-theoretic approach to non-locality and contextuality \cite{abramsky:11,abramsky:11a,mansfield:13b}. The resulting criterion has several advantages. It avoids certain technical difficulties, and due to its generality it can be applied within any ontological physical theory: e.g. generalised probabilistic theories \cite{barrett:07}, or classical mechanics.

The initial investigations here also show how such considerations can provide an alternative perspective on foundational questions more generally. We find an alternative characterisation of both local and non-contextual correlations, as those that can arise from observations or measurements of properties that can be considered `real' in this sense. This ties together the notions of locality and reality, bringing to light another link between the Bell and Pusey-Barrett-Rudolph (PBR) theorems \cite{pusey:12}, which deal, respectively, with these properties.

We begin, in section \ref{ontic}, by presenting our general, reformulated criterion for reality, which requires minimal background. Much of the literature on the foundations of quantum mechanics, including that concerning recent developments on the reality of the wavefunction, deals with hidden variable or ontological models. Therefore, we will provide a brief review of this framework in section \ref{ontsect}, which readers familiar with the material may wish to skim over, paying attention to the notation used. In section \ref{observables}, we apply the criterion not to the wavefunction but to observable properties, leading to a characterisation of locality akin to that of the unified sheaf-theoretic approach to non-locality and contextuality \cite{abramsky:11} or to Kochen-Specker contextuality \cite{kochen:75}. We demonstrate how this may be used to arrive at treatments of the fact that local hidden variable models can be subsumed by the sheaf-theoretic framework \cite{abramsky:11}, and the EPR argument \cite{einstein:35}.

Finally, in section \ref{sec:pbr}, we give a detailed consideration of preparation independence, which first appeared as one of the assumptions of the PBR theorem. We show that the assumption, which is crucial to the theorem, is analogous in a precise sense to Bell locality. Aside from this being another link between the Bell and PBR theorems, the analogy would also suggest that the assumption may be too strong, and that it could be weakened to something analogous to no-signalling \cite{ghirardi:80}. A counter-example to the PBR result was constructed by Lewis et al.~\cite{lewis:12} by dropping the assumption of preparation independence entirely, with the caveat that for compound systems it would necessarily introduce superluminal signalling. Here, we relax preparation independence to an independence assumption that is still well motivated and rules out signalling, and construct a counter-example which avoids this caveat. It is not clear, however, if it is possible to strengthen the PBR theorem so that its result still holds under the weaker assumption. We will mention, too, that by assuming preparation independence one can very easily prove Bell's result, a fact that may cast further suspicion on the strength of the stricter assumption.

\section{A Criterion for Reality}\label{ontic}

In this section we will use the terminology of Harrigan \& Spekkens \cite{harrigan:10}, which has been established in the literature. We begin by reviewing their criterion for the reality, or \emph{onticity}, of the wavefunction, which we then dualise and re-cast. We note that a dual view was suggested in \cite{harrigan:10}, though it was not formalised. For this, we need only postulate, for each system, a space $\Lambda$ of \emph{ontic states}. These can be considered to correspond to real, physical states of the system. The idea will be that objects or properties that are determined with certainty by the ontic state can themselves be considered ontic. The term ontic is chosen deliberately, and it is supposed that such objects, properties, or states, have a real objective existence as opposed to having a merely phenomenal existence. We do not, however, propose to get into a discussion of the suitability of terminology here. Similarly, objects or properties that are not determined with certainty are said to be \emph{epistemic}, recalling that the literal meaning of the term is that which relates to knowledge or to its degree of validation. The use of the term in \cite{harrigan:10} can be taken to reflect the fact that objects and properties of this kind are necessarily probabilistic and could thus be assumed to represent a degree of knowledge about some underlying ontic object or property. It should be borne in mind, of course, that results relating to these definitions will hold regardless of the physical significance attached to them.

As well as the existence of an ontic state space, the authors of \cite{harrigan:10} also posit the assumption that the preparation of any quantum state $\left| \psi \right>$ induces a distribution $\mu_{\left| \psi \right>}$ over the ontic state space $\Lambda$ for that system, specifying the probabilities for the system to be in each ontic state given that it has been prepared in this way.

\begin{definition}[Harrigan \& Spekkens \cite{harrigan:10}]\label{def:hs}
If, for all wavefunctions $\left| \psi \right> \neq \left| \phi \right>$ of each system, the induced distributions $\mu_{\left| \psi \right>}$ and $\mu_{\left| \phi \right>}$ have non-overlapping supports, the wavefunction is said to be \emph{ontic}. Otherwise, there exist some $\left| \psi \right> \neq \left| \phi \right>$ such that $\mu_{\left| \psi \right>}(\lambda)>0$ and $\mu_{\left| \phi \right>}(\lambda)>0$ for some $\lambda \in \Lambda$, and the wavefunction is said to be \emph{epistemic}.
\end{definition}

We now formalise a more general, dual version of the definition. We thereby shift from thinking of values of properties as giving probabilistic information about ontic states, to thinking of ontic states as giving probabilistic information about values of properties. As we will see, the definition can be applied to any object or property. Though the wavefunction would more usually be considered to be (at least) a mathematical object rather than a property of a system, for simplicity we only refer to properties from now on. The terms ontic and epistemic apply to the relationship or specification of values for each ontic state, and it is this that we refer to as a property.

\begin{definition}\label{def:ontepi}
A \emph{$\mathcal{V}$-valued property over $\Lambda$} is a function $f: \Lambda \rightarrow \mathcal{D}(\mathcal{V})$, where $\mathcal{D}(\mathcal{V})$ is the set of probability distributions over $\mathcal{V}$.
The property is said to be \emph{ontic} in the special case that, for all $\lambda \in \Lambda$, the distribution $f(\lambda)$ over $\mathcal{V}$ is a delta function.
Otherwise, it is said to be \emph{epistemic}.
\end{definition}

Another way of stating this is that ontic properties are generated by functions $\widehat{f}:\Lambda \rightarrow \mathcal{V}$; i.e. they map each ontic state to a unique value. For epistemic properties, however, there is at least one ontic state that is compatible with two or more distinct values in $\mathcal{V}$.

We now set about showing how these definitions relate, which may not be immediately clear. Any $\mathcal{V}$-valued property $f$ specifies probability distributions over $\mathcal{V}$, conditioned on $\Lambda$. Bayesian inversion can be used to obtain probability distributions over $\Lambda$, conditioned on $\mathcal{V}$, which we (suggestively) label $\{ \mu_v \}_{v \in \mathcal{V}}$. Explicitly,
\begin{equation}\label{distdef}
\mu_v(\lambda) := \frac{ f(\lambda) (v) \cdot p(\lambda)}{ \int_{\Lambda} f(\lambda') (v) \cdot p(\lambda') \, d\lambda'},
\end{equation}
assuming a uniform prior distribution $p(\lambda)$ on $\Lambda$. Note that this is only well-defined for finite $\Lambda$, and that a more careful measure theoretic treatment, which will not be provided here, is required for the infinite case.

\begin{proposition}\label{prop:hscor}
A $\mathcal{V}$-valued property over finite $\Lambda$ is ontic (definition \ref{def:ontepi}) if and only if the distributions $\{ \mu_v \}_{v \in \mathcal{V}}$ have non-overlapping supports.
\end{proposition}

\begin{proof}
Suppose the property $f$ is ontic according to definition \ref{def:ontepi}, let $\lambda \in \Lambda$, and let $v,v' \in \mathcal{V}$ such that $v \neq v'$. Assume for a contradiction that $\mu_v(\lambda)>0$ and $\mu_{v'}(\lambda)>0$. Then, by (\ref{distdef}), $f(\lambda) (v)>0$ and $f(\lambda) (v')>0$; but since $f$ is ontic,
\[
v_\lambda = v \neq v' = v_\lambda,
\]
where $v_\lambda := \widehat{f}(\lambda)$.

Conversely, suppose that the distributions $\{ \mu_v \}_{v \in \mathcal{V}}$ have non-overlapping supports and assume for a contradiction that $f(\lambda) (v)>0$ and $f(\lambda) (v')>0$. Then, by (\ref{distdef}), $\mu_v(\lambda)>0$ and $\mu_{v'}(\lambda)>0$.
\end{proof}

One way of thinking about this correspondence, which may merit further research, could be as a kind of Stone duality, or as a special case of the dual equivalence between the category of von Neumann algebras and $*$-homomorphisms and the category of measure spaces and measurable functions \cite{heunen:13}.

To illustrate, we provide a couple of simple examples of ontic and epistemic properties.

\begin{example}[Classical Mechanics]
The phase space of a system is taken to be the ontic state space. Classical mechanical observables (energy, momentum, etc.) are represented by real-valued functions on phase space, and are therefore ontic.
\end{example}

\begin{example}[Fuzzy Measurement]\label{ex:epi}
Consider an experiment in which a bag is prepared containing two coins, which can each be green or white, with equal probability, but are otherwise identical. We claim that the process of removing one and checking its colour measures an epistemic property. If the ontic states are $\Lambda = \{ GG, GW, WG, WW \}$, the property cannot be represented by a $\{G,W\}$-valued function on $\Lambda$. Given the ontic state $GW$, for example, both $G$ and $W$ are compatible, and can arise with equal probability.
\end{example}

In this second example, the property that is being measured is, according to the present definition, epistemic with respect to the state of the bag; it might also be said the example describes a fuzzy measurement on the state of the bag.

Definition \ref{def:ontepi} has some advantages.
\begin{itemize}
\item
It is fully general and can be applied to any object or property in any ontological theory.
\item
It avoids measure theoretic problems relating to sets of measure zero that are inherent to the original \cite{hall:11}.
\item
It is also mathematically straightforward and conceptually transparent.
\item
We will see in section \ref{observables} that, while generalisation would be possible in the original formulation, generalising the criterion in the present form avoids the need to postulate additional non-specified distributions.
\end{itemize}

\section{Ontological Models}\label{ontsect}

We are concerned with theories that give operational predictions for outcomes to measurements; we refer to sets of such predictions as empirical models. Quantum mechanics is one such theory, which might be described operationally by saying that we associate a density matrix $\rho^p$ with each preparation $p$, a POVM $\{E^m_o\}_{o \in O}$ with each measurement $m$, and prescribe the probability of the outcome $o$ given preparation $p$ and measurement $m$ by
\[
p(o\mid m,p) = \text{tr}(\rho^p \, E^m_o).
\]

We wish, more generally, to consider theories with the same kind of operational structure. In order to do so, we will use some notation that is similar to that of the sheaf-theoretic approach. For each system we assume spaces $P$ of preparations, $X$ of measurements, and $O$ of outcomes. There may be some compatibility structure on the space of measurements, say $\mathcal{M} \subseteq \mathcal{P}(X)$, specifying which sets of measurements can be made jointly (in quantum mechanics, this is specified by the commutative sub-algebras of the algebra of observables). This information encodes which kind of measurement scenario we are working in: e.g. the Bell-CHSH model \cite{clauser:69,bell:87}, Hardy model \cite{hardy:92,hardy:93}, and PR correlations \cite{popescu:94} all deal with two-party scenarios in which each party can choose freely between two binary-outcome measurements\footnote{This measurement scenario is referred to as the $(2,2,2)$ Bell scenario.}. Again, we additionally assume a space $\Lambda$ of \emph{ontic states}, over which each preparation induces a probability distribution.

In an effort to simplify notation, we will use an overline to denote a joint measurement \[\overline{m}=\{m_A,m_B, \dots \} \in \mathcal{M}\] and also to denote joint outcomes $\overline{o} \in \mathcal{E}(\overline{m})$ to a joint measurement; here $\mathcal{E}(\overline{m})$ is the set of functions $\overline{o}:\overline{m} \rightarrow O$. Readers familiar with the sheaf-theoretic approach will recall that $\mathcal{E}:X \rightarrow O^X$ is the event sheaf. On the other hand, $m\in X$ and $o \in O$, without overlines, denote individual measurements and outcomes, respectively. It may be the case that a particular individual measurement can belong to several allowed sets of joint measurements, etc. Joint preparations and joint ontic states will be treated similarly in section \ref{sec:pbr}.

\begin{definition}\label{def:hv}\index{hidden variable model}
An \emph{ontological} or \emph{hidden variable model} $h$ over $\Lambda$ specifies:
\begin{enumerate}
\item
A distribution \[h(\lambda \mid p)\] over the ontic states $\Lambda$ for each preparation $p \in P$;
\item
A distribution
\begin{equation}\label{eq:ontstatestats}
h(\overline{o}\mid \overline{m},\lambda)
\end{equation}
over joint outcomes $\E(\overline{m})$, for each ontic state $\lambda \in \Lambda$ and joint measurement $\overline{m} \in \mathcal{M}$.
\end{enumerate}
The \emph{operational probabilities}\index{operational probabilities} are then prescribed by
\begin{equation}\label{hv}
h(\overline{o}\mid \overline{m},p) = \int_\Lambda d \lambda \; h(\overline{o}\mid\overline{m},\lambda) \; h(\lambda\mid p).
\end{equation}
\end{definition}

The terms ontological model and hidden variable model are both used in the literature, but recently the term ontological model has gained some popularity. It may be a more suitable term in the sense that the `hidden' variable need not necessarily be hidden at all: it could be directly observable. In Bohmian mechanics \cite{bohm:52,bohm:52a}, for example, position and momentum play the role of the hidden variable. It also carries the connotation that such a model is an attempt to describe some underlying ontological reality.

\begin{definition}
A theory which predicts the measurement statistics for the ontic states (\ref{eq:ontstatestats}) will be referred to as an \emph{ontological theory} over $\Lambda$.
\end{definition}

We are especially interested in ontological models and theories that can reproduce quantum mechanical predictions. Trivially, the simplest such theory is quantum mechanics itself, regarded as an ontological theory.

\begin{example}[$\psi$-complete Quantum Mechanics]\index{$\psi$-completeness}
The ontic state is identified with the quantum state. A preparation produces a density matrix, which is regarded as a distribution over the projective Hilbert space associated with the system. By construction, the operational probabilities are those given by the Born rule.
\end{example}

Of course, quantum mechanics, treated as an ontological theory in itself in this way, has certain non-intuitive features (Einstein, Podolsky \& Rosen provided one early discussion of this \cite{einstein:35}) but later results such as Bell's theorem \cite{bell:64} and the Kochen-Specker theorem \cite{kochen:75} clarified the fact that non-locality and contextuality are necessary features of any theory that can account for quantum mechanical predictions. In order to address these issues, we point out some relevant properties that ontological models may have.

\begin{definition}\index{$\lambda$-independence}
An ontological model is \emph{$\lambda$-independent} if and only if the distributions over $\Lambda$ induced by each preparation $p \in P$ do not depend on the joint measurement $\overline{m} \in \mathcal{M}$ to be performed.
\end{definition}

We have already implicitly assumed this in definition \ref{def:hv}, but it is worth making it clear since it is a crucial assumption in all of the familiar no-go theorems. In a $\lambda$\emph{-dependent} model, on the other hand, the probabilities of being in the various ontic states would depend on both the preparation of the system and the joint measurement being performed, and we would have $h(\lambda\mid p,\overline{m})$ rather than $h(\lambda\mid p)$ in equation (\ref{hv}).

\begin{definition}
An ontological model is \emph{deterministic} if and only if for each $\lambda \in \Lambda$ and set of compatible measurements $\overline{m} \in \mathcal{M}$ there exists some joint outcome $\overline{o} \in \mathcal{E}(\overline{m})$ such that $h(\overline{o}\mid \overline{m},\lambda)=1$.
\end{definition}

In such a model, the outcome to any measurement that can be performed on an ontic state is determined with certainty.

For any distribution $h(\overline{o}\mid \overline{m},\lambda)$ over joint outcomes $\overline{o} \in \mathcal{E}(\overline{m})$ to the joint measurement $\overline{m} \in \M$ on the ontic state $\lambda \in \Lambda$, we can find a distribution $h(o\mid m, \lambda)$ over outcomes $o \in O$ to any individual measurement $m \in \overline{m}$ by marginalisation. 

\begin{definition}
An ontological model is \emph{parameter-independent} if and only if the probability distribution $h(o\mid m, \lambda)$ over $O$ is well-defined for each $m \in X$ and $\lambda \in \Lambda$.
\end{definition}

By well-definedness we mean that the same marginal distribution $h(o\mid m, \lambda)$ is obtained regardless of which set of joint of measurements we marginalise from (in the case that $m \in \overline{m}$ and $m \in \overline{m}'$ for example). Parameter independence thus asserts that the probabilities of outcomes to a particular measurement do not depend on the other measurements being performed. It is related to the notion of no-signalling \cite{ghirardi:80} for operational probabilities, a property which is necessarily satisfied by all quantum correlations and rules out the possibility of super-luminal signalling taking place via the measurement process at the operational level (the choice of measurement at one site cannot influence outcomes at any other site).

\begin{definition}
Operational probabilities $h(\overline{o} \mid \overline{m},p)$ are \emph{no-signalling} if and only if all marginal probability distributions $h(o \mid m,p)$ are well-defined.
\end{definition}

Preparation independence ensures no-signalling by ruling out the possibility of super-luminal causal influences at the ontological level. Precisely, preparation independence and $\lambda$-independence, together, imply no-signalling.

\begin{definition}\index{locality \& non-contextuality!ontological model/theory}
An ontological model is \emph{local (or non-contextual)} if and only if it is both deterministic and parameter-independent; empirical correlations are \emph{local (non-contextual)} if and only if they can be realised by a local (non-contextual) model.
\end{definition}

This says that for each ontic state there is a certain outcome to any measurement that can be performed, and that this does not depend on which other measurements are made. The term local is generally only used when the system being modelled is spatially distributed; where such an arrangement is not assumed, the model is said to be non-contextual.

We draw attention to the fact that another definition of locality that is common in the literature concerns the factorisability of the distributions
\begin{equation}\label{eq:bellloc}
h(\overline{o}\mid \overline{m},\lambda) = \prod_{m \in \overline{m}} \, h(o\mid m,\lambda).
\end{equation}
While the present definition may be a less familiar means of presenting non-locality, it is important to note that these definitions were shown to be equivalent, in the sense that they generate the same sets of empirical models, in \cite{abramsky:11}, which built on work by Fine \cite{fine:82} that was specific to the $(2,2,2)$ Bell scenario.

\section{Observable Properties}\label{observables}

If we are to assume that the outcomes of measurements provide the values of properties of a system, then we require that for each measurement $m \in X$ there must exist an $O$-valued property $f_m: \Lambda \rightarrow \mathcal{D}(O)$ such that $f_m(\lambda) (o) = h(o\mid m,\lambda)$ for all $\lambda \in \Lambda$ and $o \in O$.
\begin{definition}\label{def:obsprop}
The \emph{observable properties} of an ontological model $h$ over $\Lambda$ are the $O$-valued properties $f_m: \Lambda \rightarrow \mathcal{D}(O)$ given by
\begin{equation}\label{prophv}
f_m(\lambda) (o) := h(o\mid m,\lambda)
\end{equation}
for each $m \in X$ such that the marginal $h(o\mid m,\lambda)$ is well-defined.
\end{definition}
By generalising in the present dualised formulation, we avoid postulating that particular values of properties induce non-specified distributions $\mu_v$ over the space of ontic states and reasoning in terms of these, in favour of the more palatable postulate that outcomes of measurements correspond to the values of properties of a system.

\begin{theorem}\label{nogothm}\index{locality \& non-contextuality!observable properties}
An ontological model is local (or non-contextual) if and only if all measurements are of ontic observable properties.
\end{theorem}

\begin{proof}
First, we claim that a model is deterministic if and only if its observable properties are ontic. This holds since, by (\ref{prophv}),
\[
h(o\mid m,\lambda)=1 \qquad \Leftrightarrow \qquad f_m(\lambda) (o) = 1.
\]
Next, we claim that a model is parameter independent if and only if all measurements are of observable properties. This holds since, by definition \ref{def:obsprop}, all measurements are of observable properties if and only if all marginals $h(o\mid m,\lambda)$ are well-defined. The result follows.
\end{proof}

This characterisation of locality, which falls out easily from the definitions, is similar to the Kochen-Specker \cite{kochen:75} or topos approach \cite{isham:98} treatments of non-contextuality. It can provide an alternative and sometimes simpler approach to certain results. The first result we mention shows that local ontological models have a canonical form. In fact, it shows that local ontological or hidden variable models can equivalently be expressed as distributions over the set of global assignments. In this sense it shows how local ontological models are subsumed by the sheaf-theoretic approach; c.f.~main theorem of \cite{abramsky:11}, and can also be understood as a generalisation of the work of Fine \cite{fine:82}. An interesting, related point is that, by allowing for negative probabilities, these canonical models can also generate all no-signalling correlations \cite{abramsky:11,abramsky:14,mansfield:13t}.

\begin{theorem}\label{thm:canonical}
Local models can be expressed in a \emph{canonical form}, with an ontic state space $\Omega := \mathcal{E}(X)$, and probabilities
\[
h(\overline{o}\mid \overline{m},\omega) = \prod_{m \in \overline{m}} \, \delta \left( \omega(m), \overline{o}(m) \right)
\]
for all $\overline{m} \in \mathcal{M}$, $\overline{o} \in \mathcal{E}(\overline{m})$, and $\omega \in \Omega$.
\end{theorem}

\begin{proof}
See \cite{mansfield:13t}.
\end{proof}

The next proposition will not be surprising in light of the EPR argument \cite{einstein:35}. It shows that if one were to take the view that quantum mechanics is $\psi$-complete then all non-trivial observables are epistemic or inherently probabilistic. Indeed, we can obtain a re-statement of the EPR result as a corollary.

\begin{proposition}\label{incomp}\index{$\psi$-completeness}
Any non-trivial quantum mechanical observable is epistemic with respect to $\psi$-complete quantum mechanics.
\end{proposition}

\begin{proof}
Any observable $\hat{A} \neq \mathbf{I}$ has eigenvectors, say $\left|  v_1 \right>$ and $\left|v_2 \right>$, corresponding to distinct eigenvalues, say $o_1$ and $o_2$. Consider any state $\left| \psi \right>$ such that $\left< v_1 | \psi \right> >0$ and $\left< v_2 | \psi \right> >0$. In a $\psi$-complete model, the wavefunction is the ontic state, so $\lambda = \left| \psi \right>$. Then
\[
f_{\hat{A}} (\lambda) (o_1) = h(o_1\mid \hat{A},\lambda) = \left| \left< v_1 | \psi  \right> \right|^2 >0,
\]
and similarly $f_{\hat{A}} (\lambda) (o_2) >0$. Therefore $f_{\hat{A}}$ is epistemic.
\end{proof}

\begin{corollary}[EPR]\label{prebell}\index{EPR argument}
Under the assumption of locality, quantum mechanics cannot be $\psi$-complete.
\end{corollary}

\begin{proof}
By Proposition \ref{incomp}, any non-trivial quantum observable is epistemic with respect to $\psi$-complete quantum mechanics. Therefore, by Theorem \ref{nogothm}, $\psi$-complete quantum mechanics is not local.
\end{proof}

This is the same result that was argued for by EPR, though this proof has more in common with an earlier argument by Einstein at the 1927 Solvay conference \cite{bacciagaluppi:10}, and also with a more recent, general treatment found in \cite{abramsky:10} and \cite{brandenburger:08}.

\section{The PBR Theorem}\label{sec:pbr}

In this section we briefly make some observations relating to the PBR theorem, which deals with the reality (i.e. onticity in the sense of definitions \ref{def:hs} and \ref{def:ontepi}) of the wavefunction. One of the assumptions for this result is \emph{preparation independence} \cite{pusey:12}: \begin{quote} systems that are prepared independently have independent physical states.\end{quote} The other assumptions are implicit in the present framework.

\begin{theorem}[PBR]\label{pbrthm}
For any preparation independent theory that reproduces (a certain set of) quantum correlations, the wavefunction is ontic.
\end{theorem}

The preparation independence assumption is concerned with the composition of systems and has not appeared in previous no-go results. We will attempt to give this a more careful treatment. First of all, the PBR theorem describes a \emph{preparation scenario}. More generally, we might think of preparation scenarios as an analogue of measurement scenarios, in which the preparations $P$ play the role of measurements and the ontic states $\Lambda$ play the role of outcomes; see Table \ref{tab:mpco}.
Just as we had a compatibility structure $\M$ for measurements, which in Bell scenarios allowed us to chose one measurement from each site, we should in general allow for a compatibility structure $\mathcal{P}$ for preparations, which in the case of the PBR result allows us to chose one preparation per site. We should allow for joint ontic states $\overline{\lambda}$, just as we allowed for joint outcomes. Similarly to before, we will take $\overline{p}$ to denote a tuple of joint preparations, one for each site, and $\overline{\lambda}: \overline{p} \rightarrow \Lambda$ to denote a tuple of joint hidden variables. The definitions of an ontological model and the properties from section \ref{ontsect} can be modified in the obvious way to account for this additional structure.

\begin{table}
\begin{center}
\caption{\label{tab:mpco} Analogy between measurement and preparation scenarios, up to a sheaf-theoretic description of preparation models (c.f.~\cite{abramsky:11}).}
\begin{doublespace}
\begin{tabular}{lc|lc}
\multicolumn{2}{c|}{Measurement Scenario} & \multicolumn{2}{c}{Preparation Scenario} \\ \hline
Measurements & $X$ & Preparations & $P$ \\
Outcomes & $O$ & Ontic states & $\Lambda$ \\
Non-locality & ~ & Preparation independence & ~ \\
No-signalling & ~ & No-preparation-signalling & ~ \\
Measurement compatibility & $\mathcal{M}$ & Preparation compatibility & $\mathcal{P}$ \\
Measurement events & $\E(\overline{m}) := O^{\overline{m}}$ & Preparation events & $\E(\overline{p}) := \Lambda^{\overline{p}}$ \\
Empirical model & $\{e_{\overline{m}}\}_{\overline{m}\in\M}$ & Preparation model & $\{e_{\overline{p}}\}_{\overline{p}\in\mathcal{P}}$ \\
\multicolumn{2}{r|}{$\left(\Forall{\overline{m} \in \mathcal{M}} \; e_{\overline{m}} \in \mathcal{D}\E(\overline{m})\right)$} & \multicolumn{2}{r}{$\left(\Forall{\overline{p} \in \mathcal{P}} \; e_{\overline{p}} \in \mathcal{D}\E(\overline{p})\right)$}
\end{tabular}
\end{doublespace}
\end{center}
\end{table}

We are now in a position to give a more careful definition of preparation independence.

\begin{definition}
An ontological theory $h$ over $\Lambda$ is \emph{preparation independent} if and only if we can factor
\begin{equation}\label{eq:lsep}
h(\overline{\lambda} \mid \overline{p}) = \prod_{p \in \overline{p}} \, h( \lambda_p \mid p)
\end{equation}
for all $\overline{p} \in \mathcal{P}$, where $\lambda_p := \overline{\lambda}|_{p}$.
\end{definition}

Presented in this way, preparation independence (\ref{eq:lsep}) in a preparation scenario is clearly seen to be analogous to non-contextuality or Bell locality (\ref{eq:bellloc}) in a measurement scenario. An intriguing question is what happens if this is relaxed to an assumption analogous to no-signalling, in which we only assume that the marginal distributions $h(\lambda_p \mid p)$ are well-defined. Such a `no-preparation-signalling' assumption would ensure that the preparation at one site cannot affect the probabilities of various ontic states at another site. Preparation independence would trivially imply no-preparation-signalling, but not vice versa. It is true that it would allow for global or non-local correlations in the joint ontic state $\overline{\lambda}$; but perhaps in light of the Bell and Kochen-Specker theorems this is to be expected. We therefore propose this as a more reasonable independence condition.

\begin{definition}\label{def:noprepsig}
An ontological theory $h$ over $\Lambda$ is \emph{no-preparation-signalling} if and only if the marginal probabilities $h(\lambda_p \mid \overline{p})$ are well-defined.
\end{definition}

If we weaken the assumption of preparation independence to that of no-preparation-signalling, we will show it is possible to avoid the conclusion of PBR. Under the modified assumptions we will show how to construct a counter-example to the \emph{argument} given by PBR for the onticity of the wavefunction \cite{pusey:12} (Proposition \ref{prop:antipbr}). The important question that will remain to be answered, therefore, is whether, with the weaker no-preparation-independence assumption, a result similar to, or indeed counter to, that of PBR can be proved. This question will be the subject of forthcoming work by the author.

\begin{proposition}\label{prop:antipbr}
The PBR argument for the onticity of the wavefunction breaks down for ontological theories which satisfy no-preparation-signalling but not preparation independence.
\end{proposition}

\begin{proof}
We begin by summarising the PBR argument up to the point at which we can find a counter-example. It is assumed that a quantum system may be prepared in states $\ket{\psi_0}$ or $\ket{\psi_1}$, inducing distributions $\mu_0(\lambda)$ and $\mu_1(\lambda)$, respectively, over the space $\Lambda$ of ontic states. Furthermore, it is assumed for a contradiction that the supports of these distributions overlap on a region $\Delta \subseteq \Lambda$, and that \[q := \min\left\{ \int_\Delta d\lambda \; \mu_0(\lambda), \int_\Delta d\lambda \; \mu_1(\lambda) \right\} >0. \]
The argument proceeds by considering two such systems, each of which is prepared independently in either $\ket{\psi_0}$ or $\ket{\psi_1}$. Given that the systems are prepared independently, then with probability $q^2>0$ both systems have ontic states in the region $\Delta$.
It therefore follows that with this probability $q^2$ the joint ontic state is compatible with each of $\ket{\psi_0} \otimes \ket{\psi_0}$, $\ket{\psi_0} \otimes \ket{\psi_1}$, $\ket{\psi_1} \otimes \ket{\psi_0}$ and $\ket{\psi_1} \otimes \ket{\psi_1}$ (in the sense that it lies in the support of the distributions on the joint ontic state space that are induced by these quantum states). However, if the systems are only required to obey no-preparation-signalling, rather than preparation independence, it is possible to find a counter-example to this step, which we present in Table \ref{tab:cexa}.

Notice that in this hypothetical preparation model, the individual ontic states never both lie in the overlap region. Hence, a joint ontic state which is compatible with all of the aforementioned quantum states can never arise. Nevertheless, the preparation model is no-preparation signalling. For either subsystem, given that the quantum state prepared is $\ket{\psi_0}$, the probability of the ontic state being in the overlap region $\Delta$ is $q$ and the probability of being outside is $1-q$, and similarly for $\ket{\psi_1}$. So the choice of quantum state prepared in one system does not affect the ontic state in the other.

\begin{table}
\caption{\label{tab:cexa} A preparation model that satisfies no-preparation-signalling but not preparation independence, and provides a counter-example to the PBR argument. The tabular representation is analogous to the framework for Bell-type measurement scenarios introduced in \cite{mansfield:11}. We read the table as saying, if the first system is prepared in the quantum state $\ket{\psi_0}$ and the second system is prepared in state $\ket{\psi_0}$ (i.e. joint quantum state $\ket{\psi_0} \otimes \ket{\psi_0}$), the probability of both ontic states being in the respective regions $\Delta$ is $0$, etc.}
\begin{center}
\begin{doublespace}
\begin{tabular}{ccccccc}
~ & ~ & ~ & \multicolumn{4}{c}{System $2$} \\
~ & ~ & ~ & \multicolumn{2}{c}{$\ket{\psi_0}$} & \multicolumn{2}{c}{$\ket{\psi_1}$} \\
~ & ~ & ~ & \multicolumn{1}{|c}{$\Delta$} & \multicolumn{1}{c|}{$\Lambda - \Delta$} & $\Delta$ & \multicolumn{1}{c|}{$\Lambda - \Delta$} \\
\cline{3-7}
\multirow{4}{*}{System $1$} & \multirow{2}{*}{$\ket{\psi_0}$} & $\Delta$ & \multicolumn{1}{|c}{$0$} & \multicolumn{1}{c|}{$q$} & $0$ & \multicolumn{1}{c|}{$q$} \\
~ & ~ & $\Lambda - \Delta$ & \multicolumn{1}{|c}{$q$} & \multicolumn{1}{c|}{$1-2q$} & $q$ & \multicolumn{1}{c|}{$1-2q$} \\
\cline{3-7}
~ & \multirow{2}{*}{$\ket{\psi_1}$} & $\Delta$ & \multicolumn{1}{|c}{$0$} & \multicolumn{1}{c|}{$q$} & $0$ & \multicolumn{1}{c|}{$q$} \\
~ & ~ & $\Lambda - \Delta$ & \multicolumn{1}{|c}{$q$} & \multicolumn{1}{c|}{$1-2q$} & $q$ & \multicolumn{1}{c|}{$1-2q$} \\
\cline{3-7}
\end{tabular}
\end{doublespace}
\end{center}
\end{table}

\end{proof}

This improves on a counter-example to the PBR argument given by Lewis, et al.~\cite{lewis:12}, which completely dropped the assumption of preparation independence. Here, we have provided a counter-example that can apply to compound systems without invoking superluminal signalling through the preparation process, due the fact that we have maintained an assumption of no-prepearation-signalling.

Another observation, which is also pointed out in \cite{harrigan:10}, is that onticity of the wavefunction is actually inconsistent with locality. This can be demonstrated as a consequence of what Schr\"{o}dinger called \emph{steering} \cite{schrodinger:36}. If a local measurement in the basis $\{\ket{0},\ket{1}\}$ is made on the first qubit of the state
\[ \ket{\phi^+} = \frac{1}{\sqrt{2}} \left( \ket{00} + \ket{11} \right) \]
then this can be considered as a remote preparation of the second qubit in one of the states $\ket{0}$ or $\ket{1}$, and similarly for a measurement in the basis $\{ \ket{+}, \ket{-} \}$. If the second sub-system has an ontic state $\lambda$ that is independent of measurements made elsewhere, then $\lambda$ must be consistent with one state from each of the sets $\{ \ket{0}, \ket{1} \}$ and $\{ \ket{+}, \ket{-} \}$, but this contradicts the onticity of the wavefunction.

The following theorem, which we propose to think of as a weak Bell theorem, since it draws the same conclusion as Bell's theorem \cite{bell:64} but with the extra assumption of preparation independence, is an obvious consequence of this.

\begin{theorem}\label{thm:wkbell}
Quantum mechanics is not realisable by any preparation independent, local ontological theory.
\end{theorem}

\begin{proof}
This follows from the PBR theorem and the occurrence of steering in quantum mechanics (see discussion above).
\end{proof}

The ease at which this result falls out may lead us to be cautious of the strength of the preparation independence assumption.

\section{Discussion}

We have presented a more general, dualised version of the criterion for the reality or onticity of the wavefunction proposed by Harrigan \& Spekkens. Recasting the criterion in this form has been seen to give certain advantages; it avoids measure theoretic technicalities relating to sets of measure zero, is general enough to apply to any object or property in any ontological theory, and is also mathematically and conceptually straightforward. Furthermore, generalising in the present formulation avoids the need to postulate that particular values of any property induce non-specified distributions over the space of ontic states.

The obvious application of the criterion to an object or property other than the wavefunction is to the observable properties of a system. This led to a characterisation of locality and non-contextuality in terms of the nature of the observed properties. This may provide a useful tool for looking at foundational results: we have used it to obtain a short proof that local ontological models have a canonical form and to gain another perspective on the EPR argument. The characterisation is similar to the Kochen-Specker \cite{kochen:75} or topos approach \cite{isham:98}\index{topos approach} treatments of non-contextuality.

It is worth mentioning that the characterisation draws a connection between locality and onticity: these are the properties that are dealt with by the Bell and PBR theorems, respectively. A further connection was found in Theorem \ref{thm:wkbell}, which showed that a weakened version of Bell's result can be obtained by an argument that combines the PBR result with the incompatibility that arises between steering and the onticity of the wavefunction.

In relation to the PBR result itself, we have attempted to give a more careful treatment of the assumption of preparation independence, and made a concrete analogy between this property and locality/non-contextuality. It is possible to relax the assumption to one that is analogous to no-signalling, and which may still be well motivated. In this case we have provided a counter-example to the PBR argument. It improves on the counter-example provided by Lewis, et al.~\cite{lewis:12} in that it applies to compound systems, while still employing a reasonable independence condition that rules out superluminal influences. This amounts to introducing global or non-local correlations in the joint ontic state, which at least is consistent with the Bell and Kochen-Specker theorems. An open question is whether by another argument the result can be shown to hold with the relaxed assumption of no-preparation-signalling. This question will be answered comprehensively in forthcoming work by the author.

\section*{Acknowledgements}
The author thanks Samson Abramsky, Clare Horsman, Nadish de Silva and Rui Soares Barbosa for comments and discussions.

\bibliographystyle{eptcs}
\bibliography{refs2doi}

\end{document}